\documentclass{amsproc}
\usepackage{amssymb,graphicx, cancel,ulem}

\def\T3{T_3}

\def\Qcal{{\mathcal Q}}
\def\Pcal{{\mathcal P}}
\def \CC{\mathcal C}
\def \SS{\mathfrak S}
\newlength{\dinwidth}
\def\le{\left}

\def\ri{\right}
\def \P{\mathcal P}
\newlength{\dinmargin}
\def\bea{\begin{eqnarray}}
\def\eea{\end{eqnarray}}   

\def \&{\hspace{-15pt}&}

\def \d{{\mathrm d}}

\newtheorem{definition}{Definition}[section]
\newtheorem{theorem}{Theorem}[section]
\newtheorem{proposition}{Proposition}[section]
\newtheorem{corollary}{Corollary}[section]
\newtheorem{remark}{Remark}[section]

\def\be{\begin{equation}}
\def\ee{\end{equation}}
\def\ben{\begin{displaymath}}
\def\een{\end{displaymath}}
\def\baa{\begin{eqnarray}}
\def\eaa{\end{eqnarray}}

\def\ba{\begin{array}}
\def\ea{\end{array}}
\makeatletter
\@addtoreset{equation}{section}
\makeatother

\def\CV{{\mathbb V}}
\def\Tcal{{\mathcal T}}

\def\Proj{ \mathbb S}

\def\M{{\mathcal M}}

\def\L{{\mathcal L}}

\def\tr{{\rm tr}\, }
\def\phi{\varphi}

\def\C{{\mathbb C}}
\def\CP1{{\mathbb C\mathbb P}^1}

\def\Z{{\mathbb  Z}}
\def\R{{\mathbb R}}
\def\a{\alpha}
\def\g{\gamma}
\def\b{\beta}
\def\ka{\kappa}

\def\p{\partial}

\def\Ccal{{\mathcal C}}
\def\Ch{{\widehat{\mathcal C}}}

\def\Mcal{{\mathcal M}}
\def\proj{{S}}

\def\gt{{\tilde{\gamma}}}

\def\dim{{\rm dim}}

\def\gh{\widehat{g}}

\def\Qc{{\mathcal Q}}

\def\f{\frac}
\def\la{\label}

\def\Ld{{\Lambda}}  
\def\qd{Q} 
\def\Scal{{\mathcal S}} 

\def\qt{{\tilde{q}}}

\def \1{{\mathrm I}}

\def \bd{\begin{definition}}
\def\ed{\end{definition}}

\begin{document}

\title{Periods of meromorphic quadratic differentials  and Goldman bracket}

\author{D. Korotkin}
\address{Department of Mathematics and
Statistics, Concordia University\\ 1455 de Maisonneuve W., Montr\'eal, Qu\'ebec,
Canada H3G 1M8}
\email{Dmitry.Korotkin@concordia.ca}
\thanks{The author thanks Marco Bertola and Chaya Norton for numerous useful conversations. The research of the author was supported by NSERC and FQRNT}


\begin{abstract}
We study symplectic properties of monodromy map for second order  linear equation with meromorphic potential
having only simple poles  on a Riemann surface. We show that the canonical symplectic structure on the cotangent bundle $T^*\Mcal_{g,n}$ implies  the Goldman Poisson bracket on the corresponding character variety under the 
monodromy map, thereby extending the recent results of the paper of 
M.Bertola, C.Norton and the author from the case of holomorphic to meromorphic potentials with simple poles.

\end{abstract}

\maketitle

\section{Introduction}

\rm

The goal of this paper is to study symplectic aspects of monodromy map for second order linear  equation on a Riemann surface of genus $g$ with meromorphic potential 
having $n$  simple poles. We generalize results of the recent paper by M.Bertola, C.Norton and the author \cite{BKN}  where the case of holomorphic potential was treated in detail
(see also an earlier paper by S.Kawai  \cite{Kawai} for another approach to this problem).

The condition of coordinate invariance of the  equation $\phi''-U\phi=0$ on a Riemann surface $\Ccal$ of genus $g$ implies  that the function $-2U$
transforms  as a projective connection under a coordinate change while the solution $\phi$ locally transforms as $-1/2$-differential \cite{HawSch}.
Thus any meromorphic  potential $U$ with simple poles can be represented  as $-S_0/2+Q$ where $S_0$ is a 
fixed holomorphic projective connection on $\Ccal$ and $Q$ is a  meromorphic quadratic differential  with $n$ simple poles. This 
parameterization of the space of meromorphic potentials on a given Riemann surface is natural since meromorphic quadratic differentials with $n$ simple poles form the cotangent space $T^*\Mcal_{g,n}$ 
to moduli space of Riemann surfaces of genus $g$ with $n$ marked points. However, there remains a freedom in the choice of  the ``base'' holomorphic projective connection $S_0$ for a given Riemann surface. In this paper, following \cite{Kawai,BKN} we  assume that $S_0$ holomorphically depends on moduli of $\Ccal$ (in particular, this requirements rules out the use of Fuchsian projective connection as $S_0$). Furthermore, following \cite{BKN} we choose $S_0$ to be the  Bergman projective connection $S_B$ 
(up to the factor $1/6$ this projective connection is the zeroth
order term in the expansion of the canonical normalized meromorphic bidifferential $B(x,y)=d_x d_y \log E(x,y)$ near the diagonal $x=y$ \cite{Fay73}; $E(x,y)$ is the prime-form).
Therefore, we are going to study the equation written in the form
\be
\phi''+\left(\frac{1}{2}S_B+Q\right)\phi=0\;.
\la{Sint} 
\ee
The ratio $f=\phi_1/\phi_2$ of two linearly independent solutions of  (\ref{Sint}) solves the Schwarzian equation
$$\Scal(f,\xi)=S_B(\xi)+Q(\xi)\;,$$
 where $\xi$ is an arbitrary local parameter on $\Ccal$ and $\Scal$ denotes Schwarzian derivative.
The Schwarzian equation determines a $PSL(2,\C)$ monodromy representation of the fundamental group 
$\pi_1(\C\setminus\{y_i\}_{i=1}^n)$ which turns out to be liftable to an $SL(2,\C)$ representation \cite{GaKaMa}
(the lift to an $SL(2,C)$ representation is a non-trivial fact due to spinorial nature of solutions $\phi_{1,2}$).
Following \cite{BKN}, we define an $SL(2,\C)$ monodromy representation of equation (\ref{Sint}) directly, by an appropriate
change of dependent variable $\phi$.

 Denote the standard generators of the fundamental group by
$\a_1,\b_1,\dots,\a_g,\b_g,\ka_1,\dots,\ka_n$; these generators satisfy the relation $(\ka_1\dots\ka_n)\prod_{i=1}^g \a_i\b_i\a_i^{-1}b_i^{-1}=id$. Since all poles $y_i$ of the potential of equation (\ref{Sint}) are simple then  both eigenvalues of 
monodromy matrices $M_{\ka_i}$ equal to $\pm 1$ for each $i$;  all other monodromies are  $SL(2)$ matrices satisfying certain genericity conditions
\cite{GaKaMa}.

Our goal is to study  symplectic properties of the monodromy map for equation (\ref{Sint}). 
Since  monodromy matrices  depend on the choice of normalization point of solutions of (\ref{Sint}) it is natural to work with the corresponding  character variety
which we denote by  $\CV_{g,n}^0$; a point of  $\CV_{g,n}^0$ is an equivalence class of monodromy representations which differ by a simultaneous conjugation with the same matrix. The index $0$ indicates that the monodromies around poles $y_1,\dots,y_n$ are not generic: all of their eigenvalues are equal to $\pm 1$.

The main result of this paper states that the 
 canonical symplectic structure on $T^*\Mcal_{g,n}$
implies the Goldman bracket on the character variety $\CV^0_{g,n}$ under the monodromy map of equation  (\ref{Sint}), therefore generalizing the result of
\cite{BKN} to  potentials with simple poles. Moreover, similar to  \cite{BKN} and \cite{Kawai}, we prove that the same statement holds for
equation (\ref{Sint}) where the Bergman base projective connection is replaced either by Schottky or quasi-fuchsian (Bers) projective connection.

We follow the same strategy  as in \cite{BKN}; it is based on
the use of the abelian periods of the quadratic differential $Q$ as  coordinates on the underlying moduli space
$T^*\Mcal_{g,n}$. These coordinates are periods of the Abelian differential $v=\sqrt{Q}$ on the canonical covering 
$\Ch$ of $\Ccal$ (the genus of $\Ch$ equals $4g-3+n$). We call them ``abelian periods'' or simply ``periods'' to distinguish from 
 monodromy
matrices of equation (\ref{Sint}) which are also  called sometimes ``periods of quadratic differentials'' 
\cite{Tyurin1978}; following the terminology of \cite{Tyurin1978} the monodromy matrices  should probably be  rather  called ``non-abelian periods''. 
In the theory of dynamical systems the abelian periods of the quadratic differential $Q$ are known under the name of
``homological coordinates'' (see \cite{DouHub} and numerous recent papers; a substantial reference list can be found in \cite{EKZ}).

The  phase space associated to equation (\ref{Sint}) is the space of pairs (Riemann surface $\Ccal$ of
genus $g$, meromorphic quadratic differential $Q$ on $\Ccal$ with $n$ simple poles). 
This phase space is nothing but the cotangent bundle $T^*\Mcal_{g,n}$ (up to subspaces of codimension one and higher;
these subspaces contain differentials with multiple zeros). The natural symplectic structure 
on $T^*\Mcal_{g,n}$ is defined by  $\omega_{can}=\sum_{i=1}^{3g-3+n} \d p_i\wedge \d q_i$,
where $q_i$ are local coordinates on the moduli space $\Mcal_{g,n}$ while $p_i$ are corresponding momenta (coefficients in the decomposition of a
cotangent vector in the basis $\{\d q_i\}$). 

The set of holomorphic local coordinates $\{q_i\}$ on $\Mcal_{g,n}$
 can be chosen as follows. To determine locally the conformal structure of $\Ccal$ we pick (outside of hyperelliptic locus and for $g>1$) 
a set of $3g-3$ entries of the period matrix $\Omega$ of $\Ccal$; in different neighbourhoods of the moduli space these 
entries might have to be chosen differently.   The quadratic differentials corresponding to cotangent vectors $d\Omega_{jk}$ are products $v_j v_k$ of   normalized holomorphic differentials. 
An additional set of $n$ coordinates which determine the positions of punctures $\{y_i\}$ on $\Ccal$ we choose to be
$q_k=(v_i/v_j)(y_k)$ where $v_i$ and $v_j$ form a pair of   normalized holomorphic 1-forms on $\Ccal$ corresponding to some Torelli marking (these coordinates are also local: in different coordinate charts on $\Mcal_{g,n}$ one might need to choose another pair of normalized holomorphic differentials and/or  different Torelli markings). The quadratic differential corresponding
to cotangent vector $dq_k$ is the meromorphic quadratic differential (given by the formula (\ref{Qk}) below) whose only simple pole is at  $y_k$. The momenta $p_i$ are then defined to be coefficients of decomposition of an arbitrary meromorphic quadratic differential with
simple poles in the basis described above.

 An alternative set of Darboux coordinates on  $T^*\Mcal_{g,n}$ (more precisely, on the subset $\Qcal_{g,n}^0\subset T^*\Mcal_{g,n}$
which contains quadratic differentials with all simple zeros)  is given by abelian periods of the quadratic differential $Q$ which are 
defined as integrals of $v=\sqrt{Q}$ over odd part of homology group $H_1(\Ch,\R)$ of canonical covering $\Ch$. The canonical 
two-sheeted covering $\Ch$  is defined
by equation $v^2=Q$ in $T^*\Ccal$; the branch points of $\Ch$ lie at zeros and poles of $Q$ (in modern literature $\Ch$ is sometimes attributed to  Seiberg-Witten \cite{GaiTes} or Hitchin \cite{Mulase}, although this two-sheeted covering was extensively used starting from early days of Teichm\"uller theory \cite{Abikoff,DouHub}). The genus of  $\Ch$ equals $4g-3+n$; it admits a natural holomorphic involution which we denote by $\mu$.  The homology group $H_1(\Ch,\R)$ can be decomposed into direct sum $H_+\oplus H_-$ of even and odd subspaces under the action of $\mu$;
 the dimension of $H_+$ equals  $2g$ while $\dim H_-=6g-6+2n$. Choosing a symplectic basis $\{a^-_i,b^-_i\}$ in $H_-$  with the intersection matrix $a_i^-\circ b_j^-=\delta_{ij}/2$ we define periods of $Q$ by
\be
A_i=\int_{a_i^-}v\;,\hskip0.7cm B_i=\int_{b_i^-}v\;.
\ee
The intersection pairing in $H_-$ defines the natural symplectic form 
$$\omega=2\sum_{i=1}^{3g-3+n} \d A_i\wedge \d B_i$$
 which turns out to coincide 
with the canonical symplectic form $\omega_{can}$ on $T^*\M_{g,n}$ restricted to the space of quadratic differentials with simple zeros. 
Moreover, we show  that the function generating the transformation from canonical Darboux coordinates  $(p_i,q_i)$
to Darboux coordinates given by periods $(\sqrt{2}A_i,\sqrt{2}B_i)$ is given by
$$
G=\sum_{i=1}^{3g-3+n} A_i B_i
$$ 
which generalizes the formula obtained in \cite{BKN} to the case $n\neq 0$.

The symplectic form $\omega=\omega_{can}$ induces a symplectic structure on the character variety $\CV_{g,n}^0$ via the monodromy map of
equation (\ref{Sint}). We emphasize that this monodromy map essentially depends on the choice of the base projective connection.
The choice  of  Bergman projective connection $S_B$ as the base  is not unique since $S_B$ transforms non-trivially under the change of Torelli marking of $\Ccal$. The   Bers projective connection chosen as the base in \cite{Kawai} 
carries even more freedom, since it depends on a choice of a point in the Teichm\"uller space as a parameter.
Nevertheless, as it was shown in \cite{BKN}, the Poisson structure induced on the character variety is the same for the Bergman projective connection (independently of Torelli marking used) and Bers projective connection used in \cite{Kawai} (independently of the choice of the initial point in the Teichm\"uller space).

Technically, it is convenient to  work with the matrix first order equation constructed by
 introducing functions $\psi_{1,2}=\phi_{1,2}\sqrt{v}$ where $\phi_{1,2}$ are two linearly 
independent solutions of   (\ref{Sint}). Denote by $\Psi$ the Wronskian matrix of 
$\psi_1$ and $\psi_2$. The matrix $\Psi$ satisfies the first order matrix equation
\be
\d\Psi= \left(\ba{cc} 0 & v \\
uv & 0  \ea \right)\Psi\;,
\la{lsi}\ee
where the meromorphic function $u$ on $\Ccal$ is given by 
\be
u= -\f{S_B-S_v}{2Q}-1
\la{potfor}\ee
and $S_v(\xi)=\Scal(z(\xi),\xi)$ is the Schwarzian derivative of the  coordinate $z(x)=\int_{x_1}^x v$ with respect to
a local coordinate $\xi$. Notice that the coefficients $v$ and $uv$ in (\ref{lsi}) are (holomorphic and meromorphic respectively)
differentials on $\Ch$, not on $\Ccal$ itself.

The Poisson bracket between $u(z)$ and $u(\zeta)$ (assuming that 
the coordinates $z$ and $\zeta$ are independent of moduli) is given by the following expression:
\be
\f{4\pi i}{3}\{u(z), u(\zeta)\}
=\L_z\left[\int^z h(z,\zeta)dz\right]-
\L_\zeta\left[\int^\zeta h(z,\zeta)d\zeta\right]\;,
\la{pbpot}\ee
where the bimeromorphic function $h(z,\zeta)$ on $\Ccal\times \Ccal$  is given by
$h(z,\zeta)=\frac{B^2(z,\zeta)}{\qd(z)\qd(\zeta)}\;$; the differential operator
  $\L_z=\f{1}{2}\p_z^3-2u(z)\p_z-u_z(z)$  is known as the  ``Lenard's operator'' in the theory of
integrable systems. 

The computation of the Poisson bracket (\ref{pbpot}) from the fundamental Poisson bracket $\{A_i, B_j\}=\delta_{ij}/2$ 
(and, therefore, also from fundamental Poisson bracket on $T^*\Mcal_{g,n}$)
is based on variational formulas for the canonical bidifferential $B$ proved in \cite{JDG,contemp}.
The monodromy map for equation (\ref{lsi}) gives an $SL(2,\C)$ representation of $\pi_1(\Ccal\setminus \{y_i\}_{i=1}^n,x_0)$.
A technical computation originally performed in \cite{BKN} allows to find the Poisson bracket between traces of monodromy matrices of
equation (\ref{lsi}) along two  arbitrary loops $\g$ and $\gt$. The result is the Goldman's bracket
\cite{Gold84}:
\be
\{{\rm tr}M_\g,\;{\rm tr}M_\gt\}=\f{1}{2}\sum_{p\in \g\circ\gt}\nu(p)\left( {\rm tr}M_{\g_p \gt}- {\rm tr}M_{\g_p \gt^{-1}}\right)\;,
\la{Goldint}
\ee
where the monodromy matrices $M_\g,M_\gt\in SL(2,\C)$;  $\g_p \gt$ and $\g_p \gt^{-1}$ are  paths obtained by  
resolving the intersection point $p$ in two different ways (see \cite{Gold84}); $\nu(p)=\pm 1$ is the contribution of the point
$p$ to the intersection index of $\gamma$ and $\tilde{\gamma}$.

The paper is organized as follows.
In Section \ref{Cancov} we describe  main objects associated to the canonical two-sheeted covering. 
In Section \ref{Seceq} we define the $SL(2,\C)$ monodromy representation for equation (\ref{Sint}) via an appropriate
matrix reformulation. In section \ref{Varfo} we derive variational formulas for coefficients of the  matrix equation and for the monodromy matrices. In Section \ref{Canper} we show that the period coordinates are Darboux coordinates for the canonical symplectic structure on $T^*\Mcal_{g,n}$. In Section \ref{canGo} we outline the modifications which have to be made in the scheme of \cite{BKN} to 
cover the case of meromorphic differentials with first order poles. 

\section{Canonical covering of a Riemann surface}
\la{Cancov}

\subsection{Setup}

Denote the moduli space of  meromorphic quadratic differentials on Riemann surface of genus $g$ with $n$ simple poles and $4g-4+n$
simple zeros by $\Qcal^0_{g,n}$.
Here we list a few basic facts about the canonical covering of a Riemann surface determined by any quadratic differential
$Q\in \Qcal^0_{g,n}$. 
The presentation is parallel to the case of canonical covering defined by a holomorphic quadratic differential which was considered in detail in \cite{Leipzig,contemp,BKN}.

Let  $\Ccal$ be a Riemann surface of genus $g$. The  Torelli marking is a choice of canonical basis  $\{a_i,b_i\}_{i=1}^g$ in $H_1(\Ccal,\Z)$ 
with the intersection index $a_i\circ b_j=\delta_{ij}$. Let $\{v_i\}_{i=1}^g$ be the dual basis in $H^{(1,0)}(\Ccal)$
normalized by $\oint_{a_i}v_j=\delta_{ij}$; integrals of $v_i$ over $b$-cycles give the period matrix
$\Omega_{ij}=\oint_{b_i} v_j$. Introduce the canonical meromorphic bidifferential $B(x,y)=d_xd_y\log E(x,y)$, where $E$ is the prime-form (sometimes $B$ is called ``Bergman kernel'' following the paper by Hawley-Schiffer 
\cite{HawSch} of 1966,
although this bidifferential was   already used by Klein \cite{Klein}). The bidifferential $B(x,y)$ is symmetric, $B(x,y)=B(y,x)$ and its only singularity is on the diagonal: as $y\to x$
in a local coordinate $\xi$ one has
\be
B(x,y)=\left(\frac{1}{(\xi(x)-\xi(y))^2}+\frac{1}{6} S_B(\xi(x))+\dots\right)d\xi(x) d\xi(y)\;.
\la{singB}
\ee
Moreover, $B(x,y)$ is normalized by the requirement that all of its $a$-periods vanish with respect to each variable (its $b_i$-period with respect, say, to $y$-variable,
equals $ 2\pi i v_i(x)$). Therefore, $B(x,y)$ depends on Torelli marking of $\Ccal$; under the change of Torelli marking it transforms according to formula  given at page 21 of \cite{Fay73}. The term $S_B$ in (\ref{singB}) transforms as projective connection under a change of coordinate $\xi$;  it is
 called  the ``Bergman projective connection'' (this name is inherited from the ``Bergman kenel'' terminology adopted in  \cite{HawSch,Tyurin1978} although, probably, it would be  historically more appropriate to call it ``Klein's projective connection'').
The projective connection $S_B$ is important both in physics (it equals to the $(zz)$-component of the energy-momentum tensor of free bosons on a Riemann surface \cite{Sonoda}) and in mathematics (it describes variations of  a holomorphic section of determinant of Hodge vector bundle over various moduli spaces
\cite{MRL,Advances,contemp}, and also variations of the determinant of Laplacian on a Riemann surface \cite{TZ1,JDG}).

\subsection{Canonical covering} 

Let $Q\in \Qcal_{g,n}^0$ be a meromorphic quadratic differential on $\Ccal$ with $n$ simple poles (denoted by $y_1,\dots,y_n$) and $4g-4+n$ simple zeros
(denoted by $x_1,\dots, x_{4g-4+n}$). The canonical covering $\Ch$ is defined by  equation $v^2=Q$ in $T^* \Mcal_g$.
The two-sheeted covering $\pi:\Ch\to\Ccal$ is branched at all poles and zeros of $Q$; thus the total number of branch points is 
$4g-4+2n$ and the  genus of $\Ch$ equals $\gh=4g-3+n$. 
The Abelian differential $v$ is holomorphic on $\Ch$; it has zeros of order 2 (on $\Ch$!) at the  branch points $\{x_i\}_{i=1}^{4g-4+n}$ and no other zeros.
Denote the natural involution on $\Ch$ by $\mu$ and decompose the homology group $H_1(\Ch,\R)$  into direct sum of the even and odd subspaces $H_+\oplus H_-$, where 
 $\dim H_+=2g$ and $\dim H_-=6g-6+2n$. The holomorphic part of cohomology group $H^{(1,0)}(\Ch,\R)$ is similarly decomposed as $H^+\oplus H^-$ where $\dim H^+=g$ and
$\dim H^-=3g-3+n$. The differential $v$ belongs to $H^-$.

Introduce a set of generators of $H_-$ denoted by $\{a_i^-,b_i^-\}$ with the intersection index $a_i^-\circ b_j^-=\delta_{ij}/2$. The integrals of $v$ over $\{a_i^-,b_i^-\}$\;,
\be
A_i=\int_{a_i^-} v\;\hskip0.7cm B_i=\int_{b_i^-} v\;, \hskip0.7cm
i=1,\dots 3g-3+n\;,
\la{homco}
\ee
are called the (abelian) periods of the quadratic differential $Q$. 
 The periods can be used as local coordinates on the space $Q_{g,n}^0$ which are called ``homological coordinates'' in the theory of dynamical systems \cite{DouHub,EKZ}
\footnote{In physics literature these periods are denoted by $a_i$ and $a_i^D$ and are attributed to Seiberg and Witten, see for example 
\cite{GaiTes} and references therein.}.

Another way of looking at coordinates $(A_i,B_i)$ is to consider them as combinations of integrals of the differential $v$ between different branch points of $\Ch$.

The differential $v$ can be  used to introduce a special system of local coordinates on $\Ch$ and $\Ccal$. Namely, if $x$ is a point of $\Ch$ which does not coincide with 
branch points $\{x_i\}$ then the local parameter (the ``flat coordinate'') on $\Ch$ in a neighbourhood of $x$ is given by
\be
z(x)=\int_{x_1}^x v\;,
\la{defflat}
\ee
where $x_1$ is a chosen ``first'' zero of $v$.


Near branch points $\{x_i\}$ of $\Ch$ the local coordinates on $\Ch$ are given by (these coordinates are called ``distinguished''):
\be
\hat{\zeta}_i(x)=\left[\int_{x_i}^x v\right]^{1/3}\;,\hskip0.7cm i=1,\dots, 3g-3+n\;;
\ee
near $y_k$ the distinguished local coordinate is
$$
\hat{\xi}_k(x)=z(x)-z(y_k)=\int_{y_k}^x v,\hskip0.7cm i=1,\dots, n\;.
$$
The ``flat'' local coordinates on $\Ccal$ near every point except zeros or poles of $Q$ are the same as on $\Ch$ (also defined up to a sign).

On the base curve $\Ccal$  the distinguished local coordinate near point $x_i$ is given by
\be
\zeta_i(x)=\left[\int_{x_i}^x v\right]^{2/3}\;,
\la{loccox}
\ee
and near $y_k$:
\be
\xi_k(x)=\left[\int_{y_k}^x v\right]^{2}\;.
\la{loccoy}
\ee

To work with the canonical covering $\Ch$ it is convenient (following Lemma 3.1 of \cite{BKN}) to choose a special 
set of generators 
 $(\{\a_i,\b_i\}_{i=1}^g$, $\gamma_{x_1},\dots,\gamma_{x_{4g-4+n}},\gamma_{y_1},\dots\gamma_{y_n})$
of the fundamental group
$\pi_1(\Ccal\setminus\{x_i,y_i\})$ 
which satisfy the relation
\be
\left( \prod_{i=1}^{4g-4+n}\gamma_{x_i}\right)  \left(  \prod_{i=1}^n  \gamma_{y_i}    \right)\prod_{i=1}^g\a_i\b_i\a_i^{-1}\b_i^{-1}=id\;.
\la{relhomgr}
\ee

The topology of the covering $\Ch$ determines a homomorphism $h$ of $\pi_1(\C\setminus\{x_i,y_i\})$ to the symmetric group $S_2$.
In parallel  to  Lemma 3.1 of \cite{BKN} one can prove that the generators $\a,\b,\g$ can be 
chosen in such a way that 
\be
h(\a_i)=h(\b_i)=(21)\;\;\hskip0.7cm i=1,\dots, g\;,
\la{hab}
\ee
while
\be
h(\g_{x_i})=h(\g_{y_j})=id\;,\hskip0.7cm i=1,\dots,4g-4+n\;,\hskip0.7cm j=1,\dots,n\;.
\la{hgg}
\ee
Then one can cut $\Ch$ along loops representing generators $\a_i$ and $\b_i$ on both sheets of the covering; in this way one gets the two-sheeted covering 
$\Ch_0$ of the fundamental polygon $\Ccal_0$ of $\Ccal$ branched at $4g-4+2n$ points $\{x_i\}$ and $\{y_i\}$.

It is convenient to connect the branch points by the following system of branch cuts:
\be
[x_i,y_i]\;,\hskip0.7cm  i=1,\dots,n\;,
\ee
\be
[x_{n+2i-1},x_{n+2i}]\;,\hskip0.7cm  i=1,\dots,2g-2\;.
\ee

Under this choice of branch cuts $\Ch_0$ becomes a two-sheeted branched covering of the fundamental polygon
$\Ccal_0$  with branch cuts chosen as shown in Fig.\ref{Chat0}.
\begin{figure}
\centering
\includegraphics[width=90mm]{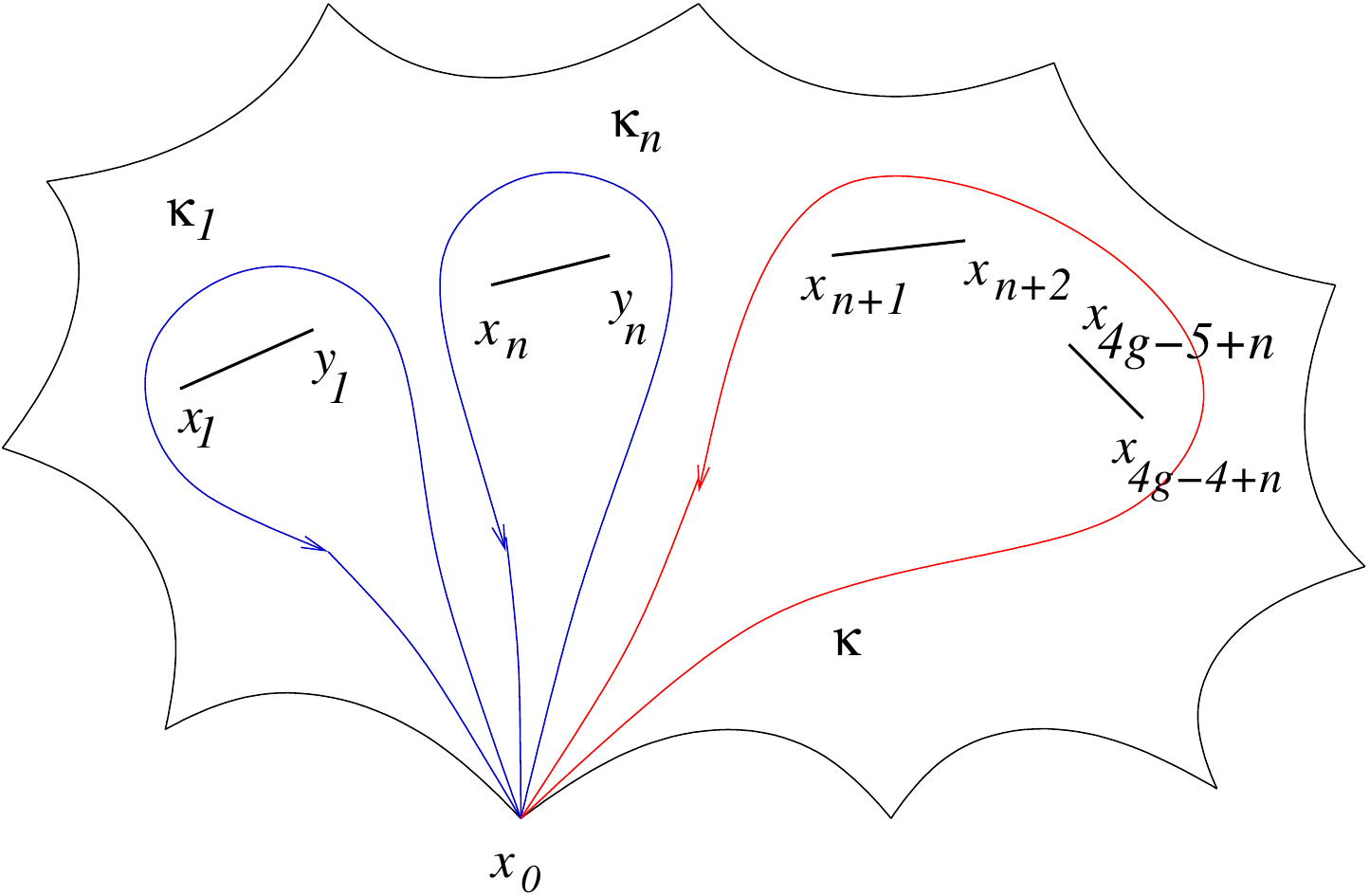}

\caption{Fundamental polygon $\Ccal_0$ with loops $\kappa_i$ and  $\kappa$.}
\la{Chat0}
\end{figure}

\hskip3.0cm

Denote by $x_0$ the corner of $\Ccal_0$. To study the monodromy group of the  equation (\ref{Sint}) 
we are going to introduce a system of $n$ loops $\kappa_1,\dots,\kappa_n$   on $\Ch_0$ (Fig.\ref{Chat0}) such that the loop $\kappa_i$ goes around the branch cut $[x_i,\,y_i]$.
Projections of these loops to $\Ccal_0$ will be denoted by the same letters.

\section{Second order  equation with meromorphic potential on a Riemann surface}

\label{Seceq}

We are going to write the second order linear equation on a Riemann surface (the "Schr\"odinger equation") in the form 
\be
\label{Schrint1}
\varphi'' +\le(\frac{1}{2}  S_B+ \qd\ri)\varphi =0 \;,
\ee 
where the derivative is taken with respect to some local coordinate $\xi$ on $\Ccal$.
The  solution $\varphi$ of (\ref{Schrint1}) is locally a $-1/2$-differential which we write as  $\phi=\phi(\xi)(d\xi)^{-1/2}$.
Choosing two linearly independent solutions of (\ref{Schrint1})  construct the Wronskian matrix 
\be
\Phi(\xi)=\left(\ba{cc}\phi_1 & \phi_2 \\
{\phi_1}_\xi       & {\phi_2}_\xi \ea\right)\;.
\la{Phidef}
\ee 
This matrix  satisfies the equation 
\be
\frac{\d\Phi}{\d\xi}= \left(\ba{cc} 0 & 1 \\ -\frac{1}{2}  S_B(\xi)- \qd(\xi) & 0 \ea\right)\Phi\;,
\la{matrixPhi}
\ee
where we use the notation $S_B=S_B(\xi) (d\xi)^2$ and  $Q=Q(\xi) (d\xi)^2$.

Let the quadratic differential $\qd$ have $n$ simple poles $\{y_i\}_{i=1}^n$ and $4g-4+n$ simple zeros  $\{x_i\}_{i=1}^{4g-4+n}$. Introduce the canonical cover $\Ch$ by
the equation $v^2=\qd$.  The zeros of the Abelian differential $v$ on $\Ch$ all have multiplicity $2$, 
and therefore $v^{1/2}$ is a section (holomorphic on $\Ch$ and
unique up to a sign)  of a spin 
line bundle over $\Ch$. 
Following \cite{BKN} we define 
\be
\psi(\xi):=\phi(\xi) v^{1/2}(\xi)\;.
\la{psiphi}
\ee
Then the scalar equation (\ref{Schrint1}) takes the following form  in terms of  $\psi$:
\be
\d\left(\f{\d\psi}{v}\right)-uv\psi=0\;,
\label{inv2}
\ee
where $u$ is the meromorphic  function on $\Ccal$ given by 
\be
u= -\f{S_B-S_v}{2v^2}-1\;.
\la{defqv}
\ee
and $S_v$ is the  meromorphic projective connection on $\Ccal$  given by the Schwarzian derivative of the flat coordinate $z$:
\be
S_v(\xi)=\Scal\left(\int_{x_1}^x v, \xi\right)=\left(\frac{v'}{v}\right)'-\frac{1}{2}\left(\frac{v'}{v}\right)^2\;.
\la{Svdef}
\ee

Choosing two linearly independent 
solutions $\psi_{1,2}$ of (\ref{inv2}) we construct the Wronskian matrix $\Psi$ which is related to
the matrix $\Phi$ (\ref{Phidef}) as follows
\be
\Psi=\left(\ba{cc}\psi_1 & \psi_2 \\
d\psi_1/v       & d\psi_2/v \ea\right)\equiv 
 v^{1/2}(\xi)\left[I +\frac{v_\xi}{2v}\left(\ba{cc} 0 & 1 \\
0   & 0 \ea\right)\right]\Phi\;.
\la{defPsi}
\ee
The matrix $\Psi$ satisfies the equation
\be
\d\Psi= \left(\ba{cc} 0 & v \\
uv & 0  \ea \right)\Psi\;.
\la{matreq}
\ee
Matrix entries $v$ and $uv$ of the coefficient matrix of (\ref{matreq}) are differentials on $\Ch$ ($v$ is holomorphic while $uv$ is meromorphic) which are anti-symmetric under the
involution $\mu$.

The advantage of the matrix $\Psi$ over the original Wronskian matrix $\Phi$ is that the matrix
entries of the  former are functions i.e. they 
are independent of the choice of the local parameter $\xi$ while the matrix entries of the latter
non-trivially depend on the choice of $\xi$.
The price which is paid for this independence is the appearance of the canonical covering $\Ch$ 
in the  equation satisfied by $\Psi$. Therefore, $\Psi$ gets additional monodromies $\pm i$ around zeros $x_i$
($\Phi$ is  monodromy-free around $x_i$). The same factors $\pm i$ appear when comparing monodromies 
of (\ref{matrixPhi}) and (\ref{matreq}) around poles $y_i$.

Another way to rewrite equation (\ref{matreq}) is to use the flat coordinate $z=\int_{x_1}^x v$
on $\Ch$ ($z$ can be used as a local coordinate on $\Ccal$ and $\Ch$ outside of neighbourhoods of points $x_i$ and $y_k$) . Then $v=dz$ and the scalar equation (\ref{inv2})    can be written as 
\be
\psi_{zz}+u\psi=0\;.
\la{linearz}
\ee
The matrix equation  (\ref{matreq}) now takes the form
\be
 \f{d\Psi}{d z}= \left(\ba{cc} 0 & 1 \\
u(z) & 0  \ea \right)\Psi\;.
\la{Psiz}
\ee

\subsection{ Monodromy representation}
\la{mondef}

There are various ways to associate a monodromy representation of
the fundamental group $\pi_1(\Ccal\setminus \{y_i\}_{i=1}^{n},x_0)$ 
to  equations (\ref{Schrint1}), (\ref{matrixPhi}) and (\ref{matreq}).
Due to the sign ambiguity of the spinors $\phi_{1,2}$ the equation for $\Phi$ (\ref{matrixPhi}) 
 determines only a $PSL(2,\C)$ monodromy representation which coincides with the 
monodromy representation of the Schwarzian equation for function $f=\phi_1/\phi_2$:
\be
\Scal (f,\xi)=S_B+2Q\;.
\la{Scheq}
\ee 
It is a non-trivial problem whether or not this representation can be lifted to an $SL(2,\C)$ representation (see \cite{GaKaMa}). 

On the other hand, if we write down the equation in terms of matrix $\Psi$ (\ref{matreq}) an
$SL(2,\C)$ representation can be naturally constructed (as shown in \cite{BKN} for $n=0$).

Let us show how an  $SL(2,\C)$ monodromy representation of 
$\pi_1(\Ccal\setminus \{y_i\}_{i=1}^{n},x_0)$ can be associated to equation (\ref{matreq}) for
$n>0$. The subtlety is in the fact that both differentials, $v$ and $uv$, forming 
the matrix of coefficients of the system, are differentials on $\Ch$, not on $\Ccal$.
However a natural $SL(2,\C)$ representation of the fundamental group of $\Ccal\setminus\{y_k\}_{k=1}^n$  can be defined as follows.

Assume that the generators of the fundamental group $\pi_1(\Ccal\setminus\{x_i,y_i\})$ are chosen according to (\ref{hab}), (\ref{hgg})
and consider the first sheet  $\Ch_0^{(1)}$ of the two-sheeted covering $\Ch_0$ of the fundamental polygon $\Ccal_0$ (Fig.\ref{Chat0}).
As before, the loops $\kappa_i$, $i=1,\dots,n$  are chosen to encircle the branch cuts $[x_i,y_i]$; introduce also the loop $\kappa$
excircling all the remaining  branch cuts $[x_{i},x_{i+1}]$ for $i=n+1,\dots,4g-4+n-1$. 

Since the coefficients of equation (\ref{matreq}) are single-valued in the multiply-connected domain obtained by deleting all branch cuts from $\Ch_0^{(1)}$,
one can define monodromy matrices along the loops $\a_i$, $\b_i$, $\kappa_i$ and $\kappa$ corresonding to initial point $x_0$.

The monodromy matrix along $\kappa$ equals $I$ 
since (see Lemma 6.2 of \cite{BKN}) the monodromy around each $x_i$ arises only due to the factor of $v^{1/2}$ relating matrices $\Phi$ (which is monodromy-free
around $x_i$) and $\Psi$; this monodromy equals $i$. Since the number of zeros encircled by $\kappa$ is a multiple of 4, the  total monodromy 
of $\Psi$ along $\kappa$ equals $I$.

Therefore the  monodromies along remaining loops $\kappa_i$, $\{\alpha_i\}$, $\{\beta_i\}$ satisfy the required relation
\be
M_{\kappa_n}\dots M_{\kappa_1}\prod_{i=g}^1 M_{a_i} M_{\b_i} M_{\a_i}^{-1}M_{\b_i}^{-1}=I\;,
\la{relmon}
\ee
i.e. these monodromies form an anti-representation of the fundamental group of $\Ccal\setminus \{[x_i,y_i]\}_{i=1}^n$.
Since  the monodromies of $\Psi$ around $x_i$  arising from elementary transformation between $\Phi$ and $\Psi$ are equal to $i I$
we shall {\it define} the monodromies of equation (\ref{Schrint1}) around poles $y_i$ to be matrices $M_{\kappa_i}$. These matrices obviously correspond to
$PSL(2,\C)$ monodromies of the associate Schwarzian equation
\be
\Scal(f,\xi)=S_B(\xi)+Q(\xi)\;.
\la{Schweq}\ee

Therefore, we get the following 
\begin{proposition}
The  constructed $SL(2,\C)$ monodromy representation of matrix equation (\ref{matreq}) is a lift of the $PSL(2,\C)$ monodromy representation of the 
Schwarzian equation (\ref{Schweq}).
\end{proposition}

The equation (\ref{Sint}) can be written in terms of  a local coordinate $\xi$ near a pole $y_i$ in the form
$$
\varphi_{\xi\xi}+\left(\frac{C}{\xi}+ O(1)\right)\varphi=0\;.
$$ 
The local analysis shows that the monodromy of $\Phi$ around $y_i$ has two coinciding eigenvalues equal to $1$.
Due to additional factors of $-i$ and $i$ arising from monodromies of $v^{1/2}$ around $y_i$ and $x_i$, respectively,  the eigenvalues 
of monodromy $M_{\ka_i}$ of $\Psi$ along $\ka_i$ are also coinciding and equal to $+1$.

\section{Variational formulas}
\la{Varfo}

\subsection{Variational formulas on $Q_{g,n}^0$}

The variational formulas on the moduli space $Q_{g,n}^0$ describe the dependence of the period matrix $\Omega$, normalized differentials $v_i$ and the canonical bidifferential $B$ on moduli, which in present setting are given by periods $(A_i,B_i)$. To write down these formulas in our present setting we introduce a set
of generators $s_1\dots,s_{6g-6+2n}$ in $H_-$ (for example one can  choose this set to coincide with the set $\{a_i^-,b_i^-\}$) and introduce the periods  $\Pcal_i=\int_{s_i}v$.
Denote by $s_1^*,\dots,s_{6g-6+2n}^*$ the set of  generators dual to $\{s_i\}$ with intersection index $s_i^*\circ s_j=\delta_{ij}$; the set of dual periods corresponding to $\{a_i^-,b_i^-\}$
is given by $\{-2b_i^-,2a_i^-\}$.

The variational formulas on the space $Q_{g,n}^0$ given below can be obtained by reduction of the variational formulas on spaces of holomorphic 
 Abelian differentials  \cite{JDG} to a subspace consisting of Riemann surfaces admitting a holomorphic 
 involution. These formulas are only a slight modification of variational formulas 
on spaces of holomorphic quadratic differentials \cite{contemp,BKN}. 
Let us introduce $g$ functions on $\Ch$:
\be
f_i(x)=\frac{v_i(x)}{v(x)}
\la{deffi}
\ee
and the function of two variables on $\Ch$
\be
b(x,y)=\f{B(x,y)}{v(x)v(y)}\;.
\la{defbxy}
\ee
The variational formula for the  period matrix $\Omega$ looks as follows:
\be
\frac{\p \Omega_{jk}}{\p \P_{s_i}} = \f{1}{2} \int_{s_i^*} f_j f_k v\;.
\label{varO}
\ee
For the holomorphic normalized differentials and for  the canonical bidifferential one has:
\be
\f{\p f_j(x)}{\p \P_{s_i}} = \f{1}{2}\int_{ s_i^*} f_j(t)\, b(t,y)v(t)
\label{varV}
\ee
and
\be
\f{\p b (x,y)}{\p \P_{s_i}} = \f{1}{4\pi i}\int_{ s_i^*} b(x,t)\, b(t,y)v(t)\;,
\label{varB}
\ee
 where $\P_{s_i} := \oint_{s_i} v$ and all derivatives are computed keeping $z(x)=\int_{x_1}^xv$ and $z(y)=\int_{x_1}^y v$ constant. The fundamental polygon of $\Ch$ used to define $z(x)$ and $z(y)$ 
must be invariant under the involtion $\mu$, such that the  coordinate $z(x)$ satisfies the relation $z(x^\mu)=-z(x)$ for $x$ lying inside of the fundamental polygon.

Notice that, although the integrand in the right-hand side of (\ref{varB}) has poles of second order at the branch points $x_i$, its residues at these points vanish since 
this integrand is anti-symmetric under $\mu$.

Integrating (\ref{varB}) with respect to variable $y$ between any two points $p_1$ and $p_2$ along a contour $l$ one gets a 
variational formula for the normalized (integrals over all $a$-periods not intersecting contour $l$ vanish) differential $W_{p_1 p_2}$ of third kind with poles at $p_{1,2}$ and residues $\pm 1$:
\be
\f{\p W_{p_1,p_2} (x)}{\p \P_{s_i}} = \f{1}{4\pi i}\int_{ s_i^*}  W_{p_1,p_2} (t)\, b(t,x)v(t)\;.
\label{varW}
\ee

Taking the limit $y\to x$ in the formula (\ref{varB})  one gets  the  variational formula for the potential $u=-\f{1}{2}\frac{S_B-S_v}{Q}-1$ (also at $z(x)$  constant):
\be
\frac{\p u(x)}{\p \P_{s_i}}=-\f{3}{8\pi i} \int_{ s_i^*} h(x,t)v(t)\;,
\label{varQ}
\ee
 where 
\be
h(x,t)=\frac{B^2(x,t)}{Q(x)Q(y)}\;.
\la{defh}
\ee
The differential $h(x,t)v(t)$ has poles of order 4 on $\Ch$ at $t=x$ and $t=x^\mu$ with non-trivial residues. Therefore, the choice of the 
class of the integration path in $H_1(\Ch\setminus \{x,x^\mu\})$ in the right-hand side of (\ref{varQ}) is important; this choice must be made in such a way that
the integration goes along parts of the boundary  of the fundamental polygon of $\Ch$; the fundamental polygon should coincide with
the one  used to define the coordinate $z(x)$ which is kept fixed in the left-hand side of (\ref{varQ}).

\subsection{Variational formulas for solution and monodromies of the second  order equation}

Let $\Psi(z)$ be solution of the
matrix equation (\ref{Psiz}) normalized at the point $z_0=z(x_0)$ via $\Psi(z_0)=I$. 

Let us   introduce the following auxiliary matrix:
\be
\Ld(x)=\Psi^{-1}(x) \sigma_- \Psi(x)\;,
\label{lambdam}
\ee
where
$
\sigma_-= \left(\ba{cc} 0 & 0 \\
1 & 0  \ea \right)$.  The matrix $\Ld$ satisfies the third order equation \cite{BKN}
\be
\Ld_{zzz}-4 u(z) \Ld_z-2 u_z \Ld=0\;.
\label{3rdor}
\ee
The components of the matrix $\Lambda$ are given by $\psi_1^2,\psi_2^2$ and $\psi_1\psi_2$; therefore (\ref{3rdor}) is nothing but a well-known third order equation for 
products of two solutions of second order   equation (\ref{linearz}).

Dependence of $\Psi(x)$ (for fixed $z(x)$) on moduli is given by the following proposition:
\begin{proposition}
The following variational formula holds for solution of (\ref{matreq}) normalized by
 $\Psi(x_0)=I$:
\be
\Psi^{-1}\frac{\p\Psi(x)}{\p\Pcal_{s_i}}\Big|_{z(x)=const}= -\frac{3}{8\pi i}\int_{x_0}^x \Lambda(x)\left[\int_{t\in s_i^*} h(x,t)v(t)\right]
\la{varPsi}
\ee
for any $s_i\in H_-$. The integration path $s_i^*$ in the r.h.s. of (\ref{varPsi}) goes along the boundary of the fundamental polygon of $\Ch$
used to define the Abelian integral $z(x)$. 
\end{proposition}

The proof of this proposition is a direct application of the variation of parameters formula for the non-homogeneus linear equation using the variational formula (\ref{varQ}) for the coefficient of the equation.

An immediate corollary of the formula (\ref{varPsi}) is the following variational formula for monodromy matrices.
\begin{corollary}
Let $\gamma$ be any element of $\pi_1(\Ccal\setminus\{y_i\}_{i=1}^n, x_0)$. Denote 
by $M_\gamma$ the monodromy matrix of equation  (\ref{matreq}) along $\gamma$ (as defined in Section \ref{mondef}). Then the following variational formulas hold:
\be
M_\g^{-1}\frac{\p M_\g}{\p\Pcal_{s_i}}=-\frac{3}{8\pi i}\int_{\g} \Lambda(x)\left[\int_{t\in s_i^*} h(x,t)v(t)\right]\;.
\la{varMgam}
\ee

\end{corollary}

\section{Canonical  symplectic structure on $T^*\Mcal_{g,n}$ via periods of $Q$}
\la{Canper}

Let $\{q_i\}_{i=1}^{3g-3+n}$ be a set of holomorphic local coordinates on $\Mcal_{g,n}$. Then for any cotangent vector $p$ we can define its coordinates $p_i$ via
$p=\sum_{i=1}^{3g-3+n} p_i \d q_i$.
Introduce the canonical symplectic structure $\omega_{can}$ on $T^*\Mcal_{g,n}$:
\be
\omega_{can}=\sum_{i=1}^{3g-3+n} \d p_i\wedge \d q_i
\la{omcan}
\ee
and the corresponding symplectic potential (the Liouville 1-form)
\be
\theta_{can}=\sum_{i=1}^{3g-3+n} p_i \d q_i\;.
\la{Lican}
\ee

Define the "homological" symplectic structure on $\Qc_{g,n}^0$ in terms of periods $(A_i,B_i)$:
\be
\omega_{hom}=2\sum_{i=1}^{3g-3+n} \d A_i\wedge \d B_i
\la{omhom}
\ee
and  its symplectic potential
\be
\la{Lihom}
\theta_{hom}=2\sum_{i=1}^{3g-3+n} A_i dB_i\;.
\ee
Notice that the form $\omega_{hom}$ is independent of the choice of basis $(a_i^-,b_i^-)$ in $H_-$ but $\theta_{hom}$ depends on this choice.
On the other hand, the 1-form $ \theta=\sum_{i=1}^{3g-3+n} A_i dB_i - B_i dA_i$ (such that $d\theta=d\theta_{hom}=\omega_{hom}$)
is also independent of a choice of symplectic basis in $H_-$.

\subsection{Local coordinate systems on $\Mcal_{g,n}$ and dual quadratic differentials}

Consider first the  case $n=0$ and exclude the hyperelliptic locus from consideration.
The natural local coordinates on the moduli space $\Mcal_g$ of unpunctured Riemann surfaces of genus $g$ can be obtained using the 
Torelli theorem stating that
the complex structure of a Riemann surface $\Ccal$ is uniquely determined by its period matrix $\Omega$. 
Therefore, in a neighbourhood of any point of $\Mcal_g$ there exists a subset $D$ of $3g-3$ matrix entries of $\Omega$ which can be chosen as local coordinates $\{q_{jk}=\Omega_{jk}\}$, $(jk)\in D$  on $\Mcal_g$. 
The cotangent vectors $d\Omega_{jk}$ to $\Mcal_g$ at a given point can be identified with holomorphic quadratic differentials $v_j v_k$:
\be
d\Omega_{jk}\hskip0.5cm \sim \hskip0.5cm v_j v_k
\ee
since variation of $\Omega_{jk}$ under an infinitesimal deformation $\delta_\mu$ of a Riemann surface defined by a Beltrami differential $\mu$
($\delta_{\mu}$ can be viewed as an element of the tangent space to $\Mcal_g$ at a given point)  is given by the Rauch variational formula
$\delta_\mu \Omega_{jk}=\int_{\Ccal} v_j v_k \mu$. 
The analog of Rauch's formulas on the space $\Qcal_g$ is given by (\ref{varO}) which also confirms the identification of the tangent vector $d\Omega_{jk}$ with holomorphic quadratic differential $v_j v_k$.

For an arbitrary $n$ we are going to consider separately the  case $g\geq 2$ and the low genus cases $g=1$ and $g=0$.

\subsubsection{Coordinates on $\Mcal_{g,n}$ for $g\geq 2$}

Local holomorphic coordinates on $\Mcal_{g,n}$ can be introduced in various ways. We are going to use the following set of coordinates:
\begin{itemize}
\item
A set of $3g-3$ entries $\Omega_{jk}$ $(jk)\in D$ of the period matrix. These coordinates  determine the complex structure of $\Ccal$. The cotangent vector $\d\Omega_{jk}$  is represented by the holomorphic quadratic differential $v_j v_k$. These holomorphic quadratic differentials
span a $3g-3$-dimensional subspace in $3g-3+n$ - dimensional cotangent space $T^*_{\Ccal} \Mcal_{g,n}$.
\item
The additional $n$ coordinates (outside, possibly, a a subset of $\M_{g,n}$ of codimension 1) we are going to choose as follows:
\be
w_k=\frac{v_i}{v_j}(y_k)\;,
\la{defqk}
\ee
where $v_i$ and $v_j$ is an arbitrary pair of normalized holomorphic differentials on $\Ccal$ such that $v_j(y_k)\neq 0$. The differential 
$v_k$ has to be chosen differently in different neighbourhoods on the moduli space. [The existence of a holomorphic differential
non-vanishing at any given point $y_k$ can be easily proved. Namely, assuming that all holomorphic differentials vanish at $y_k$ we conclude that all periods of the second kind differential $B(x,y_k)$ vanish; thus the antiderivative of $B(x,y_k)$ is a meromorphic function whose only 
simple pole is at $y_k$ which is a contradiction].  

The cotangent vector $d w_k$ can be represented by the following (generically meromorphic) differential $Q_k$ whose only pole of first order is at the marked point $y_k$:
\be
Q_k(t)=\frac{1}{4\pi i} \frac{v_i(t) v_j (y_k) -v_i(y_k) v_j(t)}{v_j^2 (y_k)} B(t,y_k) \;,
\la{Qk}
\ee
where $B$ is the canonical bimeromorphic differential. If the numerator of (\ref{Qk}) vanishes at $y_k$ then the quadratic differential $Q_k$ is
holomorphic; thus the function $w_k$ can not be used as local coordinate on the subspace (of codimension 1) of $\Mcal_{g,n}$ defined by equation
$(v_i' v_j-v_i v_j')(y_k)=0$.  For our purposes, however, it is sufficient to cover $\Mcal_{g,n}$ by coordinate charts
outside of a subspace of  codimension one since all equalities of forms we derive can be extended to
the whole space by analyticity.

To verify that (\ref{Qk}) indeed corresponds to cotangent vector $\d w_k$ we use variational formulas (\ref{varV}) to get
\be
\frac{\p w_k}{\p \Pcal_i}=\int_{s_i^*} \frac{Q_k(t)}{v(t)}\;.
\la{varwk}
\ee

   
\end{itemize}

\subsubsection{Coordinates on $\Mcal_{1,n}$ }
 Let $n\geq 2$.
The dimension of $\Mcal_{1,n}$ equals $n$. The first   coordinate $q_1$ on $\Mcal_{1,n} $ is  the period $\Omega$.
The holomorphic quadratic differential corresponding to the cotangent vector $d\Omega$ is given by $v_1^2$,
in analogy to the higher genus case.

The remaining $n-1$ coordinates $q_2,\dots,q_{n}$ are defined by integrating the normalized differential $v_1$ between poles of $Q$:
\be
q_k=\int_{y_1}^{y_{k}} v_1
\la{defuk}
\ee
for $k=2,\dots,n$. 

The meromorphic quadratic differential with simple poles at $y_1$ and $y_{k}$ which corresponds to the cotangent vector $d q_k$ equals
\be
Q_k(t)= \frac{1}{4\pi i} v_1(t) W_{y_1 y_{k}}(t)\;,
\la{Qkgen1}
\ee
where $W_{xy}(t)$ is the normalized (its $a$-period vanishes)  differential of third kind on $\Ccal$ with simple poles at $x$ and $y$ and residues $+1$ and $-1$, respectively.

The correspondence between the cotangent vector $d q_k$ and the meromorphic quadratic differential $Q_k$ can be verified by  integrating  variational formulas (\ref{varV}) for the differential $v_1$ between $y_1$ and $y_{k}$ which gives
\be
\f{\p q_k}{\p\Pcal_i}=\int_{s_i^*} \frac{Q_k}{v}\;,\hskip0.7cm k=2,\dots,n\;.
\ee

\subsubsection{Coordinates on $\Mcal_{0,n}$ }

Consider the non-trivial case $n\geq 4$; then $\dim M_{0,n}=n-3$. 
 Define  coordinates $q_k$ on  $M_{0,n}$ as follows: 
 \be
q_k= \int_{y_3}^{y_{k+3}} W_{y_1 y_2} 
\la{qk0}
\ee
for $k=1,\dots,n-3$, where as before 
\be
W_{y_1 y_2} =\left(\frac{1}{x-y_1}-\frac{1}{x-y_{k+1}}\right)dx
\ee
is the meromorphic differential of third kind on the Riemann sphere.

According to general variational formulas for the third kind differential $W_{p_1 p_2}(x)$   (\ref{varW}) the quadratic differential corresponding to cotangent vector $dq_k$ is given by
\be
Q_k=\frac{1}{4\pi i} W_{y_1 y_2} W_{y_3 y_{k+3}}\;.
\la{Qk0}
\ee

Computing the coordinates $q_k$ (\ref{qk0}) and the corresponding meromorphic differentials $Q_k$ (\ref{Qk0}) 
explicitly, we get
\be
q_k=\log\frac{(y_{k+3}-y_1)(y_3-y_{2})}{(y_3-y_1)(y_{k+3}-y_2)}
\ee
and
\be
Q_k=\frac{1}{4\pi i} \frac{(y_1-y_2)(y_3-y_{k+3})}{(x-y_1)(x-y_2)(x-y_3)(x-y_{k+3})}(dx)^2\;.
\ee

\subsection{Equivalence of canonical and  homological symplectic structures on
 $T^*\Mcal_{g,n}$}

The following theorem generalizes
Theorem 4.1 of \cite{BKN} to the case $n>0$.

\begin{theorem}\la{theo1}
\la{canonhomol}

Being restricted to 
the moduli space $\Qcal_{g,n}^0$, the canonical symplectic form (\ref{omcan}) coincides with the homological symplectic form (\ref{omhom}):
\be
\omega_{can}=\omega_{hom}\;.
\la{omom}
\ee
The canonical symplectic potential (\ref{Lican}) is expressed in terms of homological coordinates as follows:
\be
\theta_{can}= \sum_{i=1}^{3g-3+n} (A_i dB_i-B_i dA_i)\;,
\la{thetaAB}
\ee
and, therefore, the function generating the transition between canonical and  homological Darboux coordinates is given by
\be
G=\sum_{i=1}^{3g-3+n}A_i B_i\;.
\la{genfun}
\ee
\end{theorem}
\begin{proof} 
For $n=0$ the proof is given  in \cite{BKN}, Th.4.1.
 Let  $n\neq 0$. Consider  the case $g\geq 2$.
At generic point of the moduli space  $\Qcal_{g,n}^0$  we can represent the quadratic differential $Q\in \Qcal_{g,n}^0$ as a linear combination of differentials $Q_k$ (\ref{Qk}) and
holomorphic differentials $v_j v_k$ where $(jk)\in D$ for some subset $D$ of matrix entries of the period matrix $\Omega$:
\be
Q=\sum_{(jk)\in D} p_{jk} v_j v_k +\sum_{l=1}^n p_l Q_l\;,
\la{QVVQ}
\ee
or, since $Q=v^2$,
\be
v=\sum_{(jk)\in D} p_{jk} \frac{v_j v_k}{v} +\sum_{l=1}^n p_l \frac{Q_l}{v}\;.
\la{VVVQ}
\ee
Integrating this relation over cycles $a_i^-$ and $b_i^-$ we get
\be
A_i=\sum_{(jk)\in D} p_{jk}\int_{a_i^-}\frac{v_j v_k}{v}+\sum_{l=1}^n p_l \int_{a_i^-}\frac{Q_l}{v}
\la{Aipp}
\ee
and 
\be
B_i=\sum_{(jk)\in D} p_{jk}\int_{b_i^-}\frac{v_j v_k}{v}+\sum_{l=1}^n p_l \int_{b_i^-}\frac{Q_l}{v}\;.
\la{Bipp}
\ee
Taking into account variational formulas (\ref{varwk}) and (\ref{varO}) we rewrite
(\ref{Aipp}) and (\ref{Bipp}) as follows:
\be
A_i=\sum_{(jk)\in D} p_{jk}\frac{\p q_{jk}}{\p B_i}+\sum_{l=1}^n p_l \frac{\p q_l}{\p B_i}
\ee
and 
\be
B_i=-\sum_{(jk)\in D} p_{jk}\frac{\p q_{jk}}{\p A_i}-\sum_{l=1}^n p_l \frac{\p q_l}{\p A_i}\;.
\ee
Therefore,
$$
\sum_{i=1}^{3g-3+n} (A_i\d B_i-B_i \d A_i)=\sum_{(jk)\in D}\sum_{i=1}^{3g-3+n}p_{jk}\left(\frac{\p q_{jk}}{\p A_i}\d A_i+
\frac{\p q_{jk}}{\p B_i}\d B_i\right)+  \sum_{l=1}^n \sum_{i=1}^{3g-3+n}p_{jk}\left(\frac{\p q_l}{\p A_i}\d A_i+
\frac{\p q_l}{\p B_i}\d B_i\right)
$$
\be
=\sum_{(jk)\in D}p_{jk}\d q_{jk}+\sum_{l=1}^n p_l \d q_l =\theta_{can}\;.
\ee

This completes the proof of (\ref{thetaAB}) for $g\geq 2$; applying $d$-operator to this relation we get (\ref{omom}).

For  $g=0,1$ the proof is  analogous. Say, in $g=1$ case the sum of residues of Abelian differential $Q/v_1$ equals $0$; therefore, $Q$ can be represented as a linear combination of $v_1^2$ and quadratic differentials (\ref{Qkgen1}):
\be
Q=p_1 v_1^2 +\sum_{k=2}^n p_k Q_k\;.
\la{QQk}
\ee
Dividing (\ref{QQk}) by $v$, integrating the result over cycles $a_i^-$ and $b_i^-$
and using variational formulas for $\Omega$ and Abelian differentials  $W_{y_1 y_k}$ we get 
 \be
A_i=p_1 \frac{\p q_{1}}{\p B_i}+\sum_{k=2}^n p_k \frac{\p q_k}{\p B_i}
\la{Aipp1}
\ee
and
 \be
B_i=-p_1 \frac{\p q_{1}}{\p A_i}-\sum_{k=2}^n p_k \frac{\p q_k}{\p A_i}\;,
\la{Bipp1}
\ee
which again implies (\ref{thetaAB}).

Consider now the case $g=0$. In genus 0 any quadratic differential with at most simple poles at $y_1,\dots,y_n$ can be represented as a linear combination  of $n-3$ quadratic differentials  $Q_k$ (\ref{Qk0}); thus
\be
v=\sum_{k=1}^{g-3} p_k \frac{Q_k}{v}
\ee
for some coefficients $p_k$.
Furthermore, variational formulas (\ref{varW}) for $W_{y_1y_2}$ imply 
$$
A_i=\sum_{k=1}^{n-3} p_k \frac{\p q_k}{\p B_i}\;,\hskip0.7cm
B_i=-\sum_{k=1}^{n-3} p_k \frac{\p q_k}{\p A_i}
$$
leading to  (\ref{thetaAB}).
\end{proof}

\begin{remark}\rm
The statement (\ref{omom}) about coincidence of the homological and canonical symplectic structures  was contained in the paper \cite{GaiTes} written in 2012 (Section 7.3.2) in the context of meromorphic quadratic differentials with second or higher order poles.\footnote{The author thanks Joerg Teschner for pointing out this reference.} However, the 
supporting argument in  \cite{GaiTes} was given
 on the ``physics'' level of rigour. On the other hand, Theorem 4.1 of \cite{BKN} for $n=0$ and  Theorem \ref{theo1} of this paper (for any   $n$) are proved rigorously. Moreover, in our Theorem \ref{theo1} and Corollary 4.1 of \cite{BKN} we get
a stronger result: we don't only prove the coincidence of the symplectic forms but also explicitly compute the generating function (\ref{genfun}) between two systems of Darboux coordinates. The correspondence of notations used here (as well as in \cite{BKN}) and notations used in \cite{GaiTes} is as follows: the homological coordinates (the ``abelian periods'' $\{A_i,B_i\}$) of $Q$ are called in \cite{GaiTes} the ``Seibeg-Witten central charge functions''  and 
denoted by $(a_i, a_i^D)$. The generating function (\ref{genfun}) of this paper coincides with function ${\mathcal F}$ of \cite{GaiTes} called there the ``Seiberg-Witten prepotential''.
The terminology of this paper follows classical works on Teichm\"uller theory; in particular, periods of $\sqrt{Q}$
were already used in 1975 paper \cite{DouHub} devoted to the theory of Strebel differentials.
\end{remark}

\section{From canonical symplectic structure on $T^*\Mcal_{g,n}$ to Goldman bracket}
\la{canGo}

The canonical symplectic structure on $T^*\Mcal_{g,n}$ (or, equivalently, the symplectic structure (\ref{omhom}) on the space $\Qcal_{g,n}^0$)
induces the Poisson structure on the space of coefficients $u$ of the equation (\ref{matreq}).  The Poisson bracket between $u(z)$ and $u(\zeta)$ (for constant $z$ and $\zeta$) can be computed using the 
variational formula (\ref{varQ}) and Theorem \ref{canonhomol}:
$$
\{u(z),u(\zeta)\}=\frac{1}{2}\sum_{i=1}^{3g-3+n} \frac{\p u(z)}{\p A_i} \frac{\p u(\zeta)}{\p B_i}
-\frac{\p u(z)}{\p B_i} \frac{\p u(\zeta)}{\p A_i}\;,
$$
 which gives, in analogy to the proof of Proposition 4.4 of \cite{BKN}:
\be
\f{4\pi i}{3}\{u(z), u(\zeta)\}
=\L_z\left[\int^z h(z,\zeta)dz\right]-
\L_\zeta\left[\int^\zeta h(z,\zeta)d\zeta\right]\;,
\la{Poissonu}
\ee
where 
 $$
 \L_z=\f{1}{2}\p_z^3-2u(z)\p_z-u_z(z)
 $$
 is called the  ``Lenard's operator'' in the theory of KdV equation \cite{BabBer}. In the computation of the bracket (\ref{Poissonu})
it is assumed that the arguments $z$ and $\zeta$ of $u$ are independent of moduli.

The bracket (\ref{Poissonu}) does not imply Poisson brackets between monodromy matrices of the equation (\ref{matreq}) themselves since 
these monodromies depend also on the choice of the basepoint $x_0$. However, the bracket (\ref{Poissonu}) defines a Poisson structure on
the character variety $\CV_{g,n}^0$ which consists of equivalence classes of monodromy representations (moduli simultaneous  conjugation by some matrix).
We add the index $0$ to the notation  $\CV_{g,n}^0$ to indicate that each of  the monodromy matrices around poles $y_i$ have  
 coinciding eigenvalues ($1$  or $-1$).

The coordinates on   $\CV_{g,n}^0$ can be chosen to be  traces of monodromy matrices $M_\gamma$ for a sufficiently large set of loops $\gamma$.

The following proposition is rather technically tedious; it is proved in complete analogy to the case $n=0$ given in Theorems 7.2 and 7.3 of \cite{BKN}:
\begin{proposition}\la{int01}
Let $\g$ and $\gt$ be two closed contours on $\CC$.  If $\g$ and $\gt$ do not intersect then the traces of the corresponding monodromy 
matrices of equation (\ref{matreq}) Poisson-commute,
\be
\{\tr M_\g,\,\tr M_\gt\}=0\;.
\la{nonint}
\ee
 If $\g$ and $\gt$
 intersect transversally at one point $x_0$  with $\gt\circ\g=1$ then
\be
\{\tr M_\g,\,\tr M_\gt\}=\f{1}{2}(\tr M_\g M_\gt-\tr M_\g M^{-1}_\gt)\;.
\label{Gold}
\ee
\end{proposition}
 
From this proposition one can deduce the validity of Goldman's bracket for any two contours $\g$ and $\gt$ by finding a sufficiently large supply of loops on $\Ccal$ which intersect at no more than one point.

\begin{theorem}\la{mainth}
The canonical Poisson structure on $T^*\Mcal_{g,n}$ implies the Goldman's bracket for traces of monodromies of equation (\ref{matreq}),
i.e. for any two loops  $\g,\gt\in \pi_1(\CC,x_0)$ one has
\be
\{\tr M_\g,\;\tr M_\gt\}= \f{1}{2}\sum_{p\in \g\cap \gt} \nu(p) (\tr M_{\g_p\gt}-\tr M_{\g_p\gt^{-1}})\;,
\label{Goldman}
\ee  
where $\g_p\gt$ and $\g_p\gt^{-1}$ are two ways to resolve the intersection point $p$ to get two new contours  $\g_p\gt$ and $\g_p\gt^{-1}$ for each $p$; $\nu(p)=\pm 1$ is the contribution of the point $p$ to the intersection index of 
$\gamma$ and $\tilde{\gamma}$.
\end{theorem}
\begin{proof}
The proof is  parallel to the proof of Theorem 7.4 of \cite{BKN}. Namely, any two loops from the set
consisting of  the following  $g^2+2g+2gn+n$ elements of $\pi_1(\Ccal\setminus \{y_i\}_{i=1}^n,x_0)$:
\be
\SS=\{\a_i,\ \b_i,\; \ 1 \leq i \leq g\;; \;\;\; \a_i\a_j, \;\; i<j\;;\; \;\; \a_i\b_j, \; i \leq j \;;\;\;
\kappa_k, \;\kappa_k\a_i,\; \ \kappa_k\b_i,\; \; k=1,\dots,n\;,\;\; i,j=1,\dots,g\}
\label{test}
\ee
either don't intersect or intersect 
at only one point; thus (\ref{Goldman}) holds for any pair of loops $\g,\gt\in \SS$. Notice that the number of functions $\tr M_\g$ for $\g\in\SS$
is always greater than the dimension $6g-6+2n$ of the character variety $\CV_{g,n}^0$; however, it remains to show that at generic point of
 $\CV_{g,n}^0$ the differentials of $\tr M_\g$ generate the whole cotangent space to  $\CV_{g,n}^0$.

By simultaneous conjugation of all monodromies we can transform the matrix $M_{\a_g}$  to a diagonal form (generically), and the matrix $M_{\b_g}$ to a matrix whose 
fixed point (of associated M\"obius transformation) equals 1. Then from the triple of numbers $(\tr M_{\a_g},\,\tr M_{\b_g}\,,\tr M_{\a_g}M_{\b_g})$
we find matrices $M_{\a_g}$ and $M_{\b_g}$. Then to find any other matrix $M_{\a_i}$ one can use the triple of loops $(\a_i,\a_i\a_g,\a_i\b_g)$ also
contained in in the set $\SS$. 

In turn, knowing all matrices $M_{\a_i}$ as well as $M_{\b_g}$ one can reproduce the  matrices $M_{\b_2},\dots,M_{\b_{g-1}}$ since for each $j$ the
set $\SS$ contains $\b_j$ and $\a_i\b_j$ for $i=1,\dots,j$, which is sufficient to determine $M_{\b_j}$ for $j\geq 2$ (up to a binary choice).

Now we have  enough information to determine $M_{\kappa_i}$ from knowing $\tr M_{\kappa_i\a_j}$ for all $\a_j$ (actually, each of $M_{\kappa_i}$ has only one 
unknown entry. The only remaining unknown matrix is $M_{\b_1}$; this matrix can be determined knowing $\tr  M_{\b_1}$, $\tr M_{\a_1}M_{\b_1}$ and
using the relation (\ref{relmon}) as discussed in \cite{BKN}.
\end{proof}

\subsection{Admissible holomorphic sections of $\Proj_g$}

The map from the space of meromorphic quadratic differentials to the character variety $\CV_{g,n}^0$ via the monodromy map of the Schr\"odinger equation is highly non-canonical: it essentially depends on the choice of the base holomorphic section of the affine bundle of projective connections over 
an appropriate covering of the
moduli space. For example,  the Bergman projective connection used in \cite{BKN} (as well as in this paper) depends on the choice of canonical basis in homology group of $\Ccal$, and therefore 
is a holomorphic section of the affine bundle of projective connections over the Torelli space. In general, to be able to compare different holomorphic sections of the affine bundle of projective connections it's convenient to consider the affine bundle $\Proj_g$ of holomorphic projective connections over the Teichm\"uller space $\Tcal_g$.

Choosing any other holomorphic section $S_0$ of  $\Proj_g$ we could also write the Schr\"odinger equation in the form
\be
\varphi''+ \left(\frac{1}{2} S_0+Q\right)\varphi=0\;.
\la{schalt}
\ee

In \cite{BKN} we have discussed alternative ways of fixing $S_0$: the Schottky or quasi-fuchsian (Bers) projective connections
(the most common Fuchsian projective connection is not suitable in this framework since it depends non-holomorphically on moduli of $\Ccal$).
In particular, we have shown that if $S_0$ is the Schottky projective connection then the canonical symplectic structure on $T^*M_g$ 
also implies the Goldman Poisson structure on the character variety $\CV_g$. The same statement for the case when $S_0$ is given by Bers projective connection
follows from the paper by S.Kawai \cite{Kawai}.  [We recall that the Schottky projective connection is given by Schwarzian derivative $\Scal(z_S,t)$
with respect to any local coordinate $t$,
where $z_S$ is the Schottky uniformization coordinate. The Bers projective is given by the Schwarzian derivative  $\Scal(z_{\Ccal_0},t)$ where $z_{\Ccal_o}$ is the coordinate
in the fundamental domain of the Kleinian group defining simultaneous uniformization of $\Ccal$ and some ``fixed'' Riemann surface $\Ccal_0$ with anti-holomorphic complex structure, i.e. 
there exists in fact infinitely many Bers  projective connections; these projective connections are labeled
by points of Teichm\"uller space].

Therefore, all of the choices of the base projective connection $S_0$ listed above  are  equivalent from symplectic point of view, 
which inspires the following definition:

\begin{definition}\la{defad}
\rm
A holomorphic section $S_0$ of the affine line bundle $\Proj_g$ over the Teichm\"uller space $\Tcal_{g}$ is called {\it admissible} if the canonical symplectic structure on $T^*\Mcal_{g,n}$ 
implies Goldman's bracket on the character variety $\CV_{g,n}$ under the monodromy map of equation (\ref{schalt}).
\end{definition}

According to discussion of \cite{BKN}, two holomorphic sections, $S_0$ and $S_1$, of $\Proj_g$ are equivalent  iff
there exists a holomorphic function $G_{01}$ on $\Tcal_g$ such that 
\be
\delta_\mu G_{01}=\int_{\Ccal} \mu (S_0-S_1)\;,
\ee
where $\mu$ is an arbitrary Beltrami differential.
The function $G_{01}$   is  the generating function of the corresponding symplectomorphizm $Q\to Q+\frac{1}{2}(S_0-S_1)$ of $T^*\Mcal_g$ \cite{BKN}.
Clearly, the equivalence of $S_0$ and $S_1$  does not depend on the number of punctures $n$. Since the equivalence of Bergman, 
Schottky and Bers projective connections follows from the analysis of $n=0$ case contained in \cite{Kawai,BKN}, we can formulate the following
corollary of Theorem \ref{mainth}. 
\begin{corollary}
If the holomorphic section of $\Proj_g$ in equation (\ref{schalt}) is chosen to be either Bergman (corresponding to any Torelli marking), Schottky (corresponding to any choice of generators) or 
 Bers (corresponding to any ``base'' Riemann surface $\Ccal_0$) projective connections
then the canonical symplectic structure on $T^*\M_{g,n}$ induces the Goldman bracket on the character variety $\CV_{g,n}^0$ under the monodromy map of
equation (\ref{schalt}).

\end{corollary}

The Definition \ref{defad} allows to formulate our Theorem \ref{theo1}  simply as follows:
\vskip0.3cm
{\it ``The Bergman projective connection is admissible for any $n$ and $g$''}.
\vskip0.3cm
Actually, the set of admissible holomorphic sections of $\Proj_g$ is rather tiny; one of examples of non-admissible holomorphic section
of $\Proj_g$ is given in Remark 5.1 of \cite{BKN}.

Suppose now that $S_0$ and $S_1$ are two admissible holomorphic sections of $\Proj_g$ and write down the same equation in two ways:
\be
\phi''- \left(\frac{1}{2}S_0 +Q_0\right)\phi=0
\la{eq1}
\ee
and
\be
\phi''- \left(\frac{1}{2}S_1 +Q_1\right)\phi=0\;,
\la{eq2}
\ee
where $Q_0$ and $Q_1$ are two meromorphic  quadratic differentails with simple poles related by
\be
Q_0-Q_1=\frac{1}{2}(S_1-S_0)\;.
\ee 
Then the canonical coverings $\Ch_0$ defined by $v_0^2=Q_0$ and $\Ch_1$ defined by   $v_1^2=Q_1$ have different conformal structure.  Moreover, the  periods
 $(A_i^0,B_i^0)$ of $v_0$ and  $(A_i^1,B_i^1)$ of $v_1$ are not related to each other in any simple way (even if $S_0$ and $S_1$ are 
Bergman projective connections corresponding to two different Torelli markings). However, since both $S_0$ and $S_1$ are admissible, the Goldman bracket on 
the character variety of equation (\ref{eq1}),  (\ref{eq2}) implies the coincidence of the homological symplectic forms defined by $Q_0$ and $Q_1$:
\be
\sum_{i=1}^{3g-3+n} dA^0_i\wedge d B^0_i = \sum_{i=1}^{3g-3+n} dA^1_i\wedge d B^1_i
\la{homhom}
\ee
i.e. each admissible projective connection defines its own set of Darboux coordinates for  Goldman's bracket. We don't know how the relation 
(\ref{homhom}) can be verified directly, without using the link with   Goldman's brackets.

\subsection{An analog of Jimbo-Miwa tau-function for the system (\ref{lsi})}

  In the theory of isomonodromic deformations 
of a linear system 
$$\frac{\d \Psi}{\d x}=\sum_{i=1}^N \frac{A_i}{x-x_i}\Psi$$
 the main role is played by the isomonodromic Jimbo-Miwa tau-function defined by equations
\be
\frac{\p\log\tau_{JM}}{\p x_i}=\frac{1}{2}{\rm res}|_{x_i} \frac{\tr (\d \Psi \Psi^{-1})^2}{\d x}\;.
\la{iso}
\ee

A straightforward analog of the definition (\ref{iso}) in the  context of equation (\ref{lsi}):
\be
\d\Psi= \left(\ba{cc} 0 & v \\
uv & 0  \ea \right)\Psi\;,
\la{lsi2}\ee
where
\be
uv= -\f{S_B-S_v}{2v}-v\;,
\la{potfor2}\ee
looks as follows:
\be
\frac{\p\log\tau}{\p \Pcal_{s_i}}:=\frac{1}{4\pi i }\int_{s_i^*} \left(\frac{\tr (\d \Psi \Psi^{-1})^2}{v} +2 v\right)=\frac{1}{2\pi i }\int_{s_i^*}(u+1)v\;,
\la{deftau}\ee
where the cycles $\{s_i,s_i^*\}$ form a symplectic basis in the odd part $H_-$ of $H_1(\Ch)$; the term $2v$ is added to provide compatibility of equations (\ref{deftau}). We notice that the addition of the analogous  term proportional to $\d x$ to the formula (\ref{iso}) does not change the right-hand side since all residues of $\d x$ vanish.
In the case of (\ref{deftau}) the addition of this counter-term is crucial.

Using the form (\ref{potfor2}) of the  potential $u$ the equations (\ref{deftau}) can be equivalently rewritten as follows in terms of the Bergman projective connection:
\be
\frac{\p\log\tau}{\p \Pcal_{s_i}}=-\frac{1}{4\pi i }\int_{s_i^*}\frac{S_B-S_v}{v} \;.
\la{defBerg}
\ee

The tau-function defined by  (\ref{defBerg}) equals the 1/8th power of the Bergman tau-function $\tau_B$ on the moduli space of quadratic differentials studied in
\cite{contemp}:
\be
\tau=\tau_B^{1/8}\;.
\ee

Alternatively the definition (\ref{defBerg}) of the function $\tau$
can be written as follows:
\be
\d\log \tau_B= \f{1}{6\pi i}\sum_{i=1}^{3g-3+n} (\tilde{A}_i\d B_i-\tilde{B}_i\d A_i)\;,
\la{altdef}
\ee
where 
\be
\tilde{A}_i=\int_{a_i^-} (u+1)v\;\hskip0.7cm
\tilde{B}_i=\int_{b_i^-} (u+1)v
\ee
and $(A_i,B_i)$ are, as before, the  periods of the differential $v$ over cycles $a_i^-$ and $b_i^-$.

The function $\tau$ is holomorphic and non-singular if zeros and poles of $Q$ don't merge.
It satisfies the following two main properties:
\begin{itemize}
\item
Transformation under change of Torelli marking (used to define the reference projective connection $S_B$) by 
an $Sp(2g,\Z)$ matrix $\sigma=\left(\begin{array}{cc} A & B\\ C & D\end{array}\right)$:
\be
\tau^\sigma= \kappa  {\rm det}^6 (C\Omega+D)\,\tau\;,
\ee
where $\kappa$ is an 8th root of unity\;.
\item
Homogeneity property:
\be
\tau(\epsilon Q, \Ccal) = \epsilon^{\frac{2}{9}(5g-5+n)}\tau( Q, \Ccal)\;,
\ee
which implies in particular the following relation between periods $(A_i,B_i)$ of the differential $v$ and the periods 
$(\tilde{A}_i,\tilde{B}_i)$ of the differential $(u+1)v$:
\be
\sum_{i=1}^{3g-3+n} (\tilde{A}_i B_i-\tilde{B}_i A_i)= \frac{2\pi i}{3} (5g-5+n)\;.
\ee
\end{itemize}

Variational formulas on spaces of quadratic differentials which are the main analytical tool for computation of Poisson brackets in this paper and in \cite{BKN}
were developed in \cite{JDG,contemp} with the goal of studying the properties of the  Bergman tau-function $\tau_B$;  
the function $\tau_B$ plays an important role in various areas - from the theory of isomonodromic deformations, random matrices and the theory  of Frobenius manifolds to spectral theory.

We expect the function $\tau$ defined in present context to be closely related to the so-called ``Yang-Yang''  function $F$ of \cite{NeRoSa} which is the  
generating ufnction between two  systems of Darboux coordinates - the first is the natural system of Darboux coordinates on $T^*\Mcal_{g,n}$,
and the second is the system of complex Fenchel-Nielsen coordinates on the character variety. In particular, as it was noticed in \cite{BKN}, the function $F$
transforms under a change of Torelli marking of $\CC$ in the same was as $\log\,\tau$. A complete elucidation of the link between $\tau$ and $F$ 
remains a challenging problem.

\section{Riemann sphere with four marked points}

Assume that the poles of $Q$ are placed to $0$, $1$, $\infty$ and $t$ in coordinate $x$.  Then the differential $Q$ can be written as
\be
Q=\frac{\mu}{x(x-1)(x-t)}(dx)^2\;.
\la{genus0}
\ee
As we see, the moduli space $\Qcal_{0,4}^0$ has dimension 2;  $\log t$ can be chosen as  coordinate
on $\Mcal_{0,4}$. To find the corresponding momentum we  choose in (\ref{qk0}) $y_1=0$, $y_2=\infty$, $y_3=1$ and $y_4=t$; then the corresponding quadratic differential $Q_1$ (\ref{Qk0}) equals
$$Q_1=\frac{1}{4\pi i}\frac{1-t}{x(x-1)(x-t)}(dt)^2\;.$$
The comparison with (\ref{genus0}) shows that the momentum $p_1$ equals to
\be
p_1=4\pi i \frac{\mu}{1-t}\;,
\ee
and, therefore, the canonical Poisson bracket $\{p_1,q_1\}=1$ implies 
$$
\{\mu, t\}=\frac{t(1-t)}{4\pi i}\;.
$$

The canonical bidifferential in terms of coordinate $x$ is written as
\be
B(x,y)=\frac{dx dy}{(x-y)^2}\;.
\la{B0}
\ee
Therefore, in the $x$-coordinate the Bergman projective connection $S_B$ is identically vanishing and the equation (\ref{Sint})
takes the form
\be
\varphi''+\frac{\mu}{x(x-1)(x-t)}\phi=0\;,
\la{4poles}
\ee
which is a special case of Heun's equation \cite{Heun}.
 
The canonical covering $\Ch$ is the elliptic curve
\be
w^2=x(x-1)(x-t)
\la{ellc}
\ee
and the differential $v$ on $\Ch$ is given by
\be
v=\frac{\sqrt{\mu}}{\sqrt{x(x-1)(x-t)}}dx\;.
\la{vgen0}
\ee

Choosing two basic cycles $a$ and $b$ on $\Ch$ (say, $a$ goes around $0$ and $1$ and $b$ goes around $1$ and $t$) we get homological coordinates 
as elliptic integrals of first kind:
\be
A=\sqrt{\mu}\int_a\frac{dx}{\sqrt{x(x-1)(x-t)}}\;,\hskip0.7cm 
B=\sqrt{\mu}\int_b\frac{dx}{\sqrt{x(x-1)(x-t)}}\;.
\ee

The equation (\ref{4poles}) has 4 monodromy matrices $M_0,M_1,M_\infty$ and $M_t$; all of their eigenvalues equal to 1. We can choose the 
generators $\ka_0,\ka_1,\ka_\infty$ and $\ka_t$ of the fundamental group  of the 4-punctured sphere such that the matrices $M_i$ satisfy the relation
$$
M_0 M_1 M_\infty M_t=I\;.
$$

The  equations (\ref{varMgam}) give derivatives of monodromy matrices with respect to  $A$ and $B$; in present case $h(x,t)=(x-t)^{-4}$.

It is also easy to derive directly the formulas for derivatives of any monodromy matrix $M_\g$ corresponding to a basepoint $x_0$ with respect to the natural coordinates $(t,\mu)$ on the space of differentials (\ref{genus0}).
Namely, introduce a Wronskian matrix $\Phi$ of two linearly independent solutions $\phi_{1,2}$ of (\ref{4poles}) 
satisfying the initial condition $\Phi(x_0)=I$. This matrix solves the equation
\be
\frac{\d\Phi}{\d x}=\left(\ba{cc} 0 & 1\\
                                -Q & 0\ea\right)\Phi\;.
\la{Phieq}
\ee
Introduce cuts connecting $x_0$ with the singular points $0,1,\infty,t$.  
Differentiating (\ref{Phieq}) with respect to $t$ and $\mu$ and using the variation of parameters formula for 
non-homogeneous equations for $\Phi_t$ and $\Phi_\mu$ we get the following result:
\be
\frac{ \d\Phi(x)}{\d t}=-\mu \Phi(x)\int_{x_0}^x  \frac{\Lambda(y)dy}{y(y-1)(y-t)^2}
\la{varPhit}
\ee
and
\be
\frac{ \d\Phi(x)}{\d\mu}=- \Phi(x)\int_{x_0}^x  \frac{\Lambda(y)dy}{y(y-1)(y-t)}\;,
\la{varPhimu}
\ee
where the integration contour is not supposed to intersect the branch cuts and, 
as before, $\Lambda(y)=\Phi^{-1}(y)\sigma_-\Phi(y)$.

Closing the integration contour along some loop $\gamma$ we get  derivatives of any monodromy matrix $M_\g$ 
with respect to $t$ and $\mu$:
\be
\frac{ \d M_\g}{\d t}=-\mu M_\g\int_{\g}  \frac{\Lambda(y)dy}{y(y-1)(y-t)^2}
\la{varMt}
\ee
and
\be
\frac{\d M_\g}{\d\mu}=- M_\g\int_{\g} \frac{\Lambda(y)dy}{y(y-1)(y-t)}\;.
\la{varMmu}
\ee

\end{document}